\documentclass[onecolumn]{article}

\usepackage{url}

\usepackage{authblk}

\usepackage{multibib}
\usepackage{amsmath,amsfonts,amsthm,amssymb,thmtools}
\usepackage{bm}
\usepackage{natbib}
\usepackage{graphicx}
\usepackage{float}
\usepackage{placeins}
\usepackage{comment}
\usepackage{tikz}
\usepackage{booktabs}
\newcommand\independent{\protect\mathpalette{\protect\independenT}{\perp}}
\def\independenT#1#2{\mathrel{\rlap{$#1#2$}\mkern2mu{#1#2}}}
\DeclareMathOperator*{\argmin}{arg\,min}

\newcommand{\yt}{{y^t}}
\newcommand{\yc}{{y^c}}
\newcommand{\yT}[1]{{y^t_{#1}}}
\newcommand{\yC}[1]{{y^c_{#1}}}
\newcommand{\yti}{\yT{i}}
\newcommand{\yci}{\yC{i}}

\newcommand{\ycj}{\yC{j}}
\newcommand{\tg}{\mathcal{T}}
\newcommand{\cg}{\mathcal{C}}
\newcommand{\nt}{{n_{t}}}
\newcommand{\nc}{{n_{c}}}
\newcommand{\tbar}{\bar{\tau}}
\newcommand{\ytbar}{{\bar{y}^t}}
\newcommand{\ycbar}{{\bar{y}^c}}
\renewcommand{\Pr}{\mathbb{P}}
\newcommand{\var}{\mathbb{V}}
\newcommand{\varhat}{\hat{\var}}
\newcommand{\syc}{S^2(Y_{\cg})}
\newcommand{\syt}{S^2(Y_{\tg})}
\newcommand{\src}{S^2(R_{\cg})}
\newcommand{\srt}{S^2(R_{\tg})}

\newcommand{\bx}{{\bm{x}}}
\newcommand{\xt}{{\tilde{\bm{x}}}}
\newcommand{\bxi}{{\bm{x}_i}}
\newcommand{\xti}{{\tilde{\bm{x}}_i}}
\newcommand{\xtj}{{\tilde{\bm{x}}_j}}

\newcommand{\predc}{\hat{y}^c(\cdot)}
\newcommand{\predt}{\hat{y}^t(\cdot)}
\newcommand{\predcx}{\hat{y}^c(\bm{x}_i)}
\newcommand{\predtx}{\hat{y}^t(\bm{x}_i)}
\newcommand{\predci}{\hat{y}^c_{-i}(\cdot)}
\newcommand{\predti}{\hat{y}^t_{-i}(\cdot)}
\newcommand{\predcxi}{\hat{y}^c_{-i}(\bm{x}_i)}
\newcommand{\predtxi}{\hat{y}^t_{-i}(\bm{x}_i)}
\newcommand{\predcxri}{\hat{y}^c_{-i}(x^r_i)}
\newcommand{\predtxri}{\hat{y}^t_{-i}(x^r_i)}
\newcommand{\predcxriLS}{\hat{y}^{c,\mathrm{LS}}_{-i}(\xti)}
\newcommand{\predcxriRF}{\hat{y}^{c,\mathrm{RF}}_{-i}(\xti)}
\newcommand{\predcxriEN}{\hat{y}^{c,\mathrm{EN}}_{-i}(\xti)}
\newcommand{\predcxrijLS}{\hat{y}^{c,\mathrm{LS}}_{-i,j}(\xtj)}
\newcommand{\predcxrijRF}{\hat{y}^{c,\mathrm{RF}}_{-i,j}(\xtj)}

\newcommand{\predr}{x^r}
\newcommand{\predri}{\predr_i}
\newcommand{\predrj}{\predr_j}
\newcommand{\predrx}{\hat{y}^r(\bm{x})}
\newcommand{\predrxi}{\hat{y}^r(\bm{x}_i)}
\newcommand{\predrfun}{\hat{y}^r(\cdot)}

\newcommand{\rti}{{r^t_i}}
\newcommand{\rci}{{r^c_i}}

\newcommand{\thim}{\hat{\tau}_i}
\newcommand{\thm}{\hat{\tau}}

\newcommand{\tls}{\hat{\beta}}
\newcommand{\tsd}{\hat{\tau}^{\mathrm{DM}}}
\newcommand{\thipw}{\hat{\tau}^{\mathrm{IPW}}}
\newcommand{\tss}{\hat{\tau}^{\mathrm{SS}}}
\newcommand{\trc}{\tss[\predr, \mathrm{LS}]}
\newcommand{\trcprf}{\tss[\xt, \mathrm{RF}]}
\newcommand{\trcpen}{\tss[\xt, \mathrm{EN}]}
\newcommand{\trebar}{\hat{\tau}^{\mathrm{RE}}}
\newcommand{\hmi}{\hat{m}_i}
\newcommand{\EE}{\mathbb{E}}

\newcommand{\Eth}{\hat{E}_{t}^2}
\newcommand{\Ech}{\hat{E}_{c}^2}

\newcommand{\tgr}{\hat{\tau}^{\mathrm{GR}}(b)}

\newcommand{\aci}{a_{-i}^{c}}
\newcommand{\ati}{a_{-i}^{t}}
\newcommand{\bci}{b_{-i}^{c}}
\newcommand{\bti}{b_{-i}^{t}}

\newcommand{\gci}{\gamma_{i}^{c}}

\newcommand{\pto}{\overset{p}{\to}}
\newcommand{\rgr}{R^{\mathrm{GR}}}
\newcommand{\srgrc}{S^2(\rgr_{\cg})}
\newcommand{\srgrt}{S^2(\rgr_{\tg})}
\newcommand{\Ect}{\tilde{E}_{c}^2}
\newcommand{\predcxrit}{\tilde{y}^c(x^r_i)}

\newcommand{\LOOP}{$\tss[\bx; \mathrm{RF}]$}
\newcommand{\ReLOOP}{$\trc$}

\newcommand{\ReLOOPEN}{$\trcpen$}

\usepackage[letterpaper,margin=1in]{geometry}

\begin{document}

  \title{\huge Precise unbiased estimation in randomized experiments using auxiliary observational data}


  \author[1]{J. A. Gagnon-Bartsch\thanks{These authors contributed equally to the paper.}}
  \author[2]{A. C. Sales$^*$}
  \author[3]{E. Wu}
  \author[4]{A. F. Botelho}
  \author[5]{J. A. Erickson}
  \author[6]{L. W. Miratrix}
  \author[7]{N. T. Heffernan}

  \affil[1]{University of Michigan, Department of Statistics; E-mail: johanngb@umich.edu}
  \affil[2]{Worcester Polytechnic Institute, Department of Mathematical Sciences; E-mail: asales@wpi.edu}
  \affil[3]{University of Virginia, Biocomplexity Institute, Social and Decision Analytics Division}
  \affil[4]{University of Florida, College of Education}
  \affil[5]{Western Kentucky University, Analytics and Information Systems} 
  \affil[6]{Harvard University, Graduate School of Education}
  \affil[7]{Worcester Polytechnic Institute, Department of Computer Science}

\date{}

\maketitle

  \begin{abstract}
{Randomized controlled trials (RCTs) admit unconfounded design-based inference---randomization largely justifies the assumptions underlying statistical effect estimates---but often have limited sample sizes. However, researchers may have access to big observational data on covariates and outcomes from RCT non-participants. For example, data from A/B tests conducted within an educational technology platform exist alongside historical observational data drawn from student logs. 
We outline a design-based approach to using such observational data for variance reduction in RCTs. First, we use the observational data to train a machine learning algorithm predicting potential outcomes using covariates, and use that algorithm to generate predictions for RCT participants. Then, we use those predictions, perhaps alongside other covariates, to adjust causal effect estimates with a flexible, design-based covariate-adjustment routine. In this way there is no danger of biases from the observational data leaking into the experimental estimates, which are guaranteed to be exactly unbiased regardless of whether the machine learning models are ``correct'' in any sense or whether the observational samples closely resemble RCT samples. We demonstrate the method in analyzing 33 randomized A/B tests, and show that it decreases standard errors relative to other estimators, sometimes substantially.
}
\end{abstract}

\section{Introduction}
Randomized controlled trials (RCTs) are famously free of confounding bias.  Indeed, a class of estimators, often referred to as ``design-based'' \citep{schochet2015statistical} or ``randomization based'' \citep{rosenbaum:2002a}, estimate treatment effects without assuming any statistical model other than whatever is implied by the experimental design itself.  Design-based statistical estimators are typically guaranteed to be unbiased.  Their associated inference---standard errors, hypothesis tests, confidence intervals---also come with accuracy guarantees.  In many cases, these apply regardless of the sample size and require only very weak regularity conditions.

While RCTs can reliably provide unbiased estimates, they are often limited in terms of precision.  The statistical precision of RCT-based estimates is inherently limited by the RCT's sample size, which itself is typically subject to a number of practical constraints.

In contrast, large observational datasets can frequently be brought to bear on some of the same questions addressed by an RCT.  Analysis of observational data, unlike RCTs, typically requires a number of untestable modeling assumptions, chief among them the assumption of no unmeasured confounding.  Consequently, treatment effect estimates from observational data cannot boast the same guarantees to accuracy as estimates from RCTs.  That said, in many cases they boast a much larger sample---and, hence, greater precision---than equivalent RCTs.  

In many cases, observational and RCT data coexist within the very same database.  For instance, covariate and outcome data for a biomedical RCT may be drawn from a database of electronic health records, and that same database may contain equivalent records for patients who did not participate in the study and were not randomized.  Along similar lines, covariate and outcome data for an RCT designed to evaluate the impact of an educational intervention might be drawn from a state administrative database, and that database may also contain information on hundreds of thousands of students who did not participate in the RCT. We refer to these individuals, who are non-participants of the RCT but who are in the same database, as the \emph{remnant} from the study \citep{rebarPaper}.
We ask, how can we use the remnant to improve power to detect effects in RCTs?

An example from the field of education is \url{www.ETRIALStestbed.org} (formerly the ASSISTments TestBed \citep{heffernan2014assistments, ostrow2016assessment}).  
The TestBed is an A/B testing program designed for conducting education research that runs within ASSISTments, and has been made accessible to third-party education researchers.  Using the TestBed, a researcher can propose A/B tests to run within ASSISTments.  That is, a researcher may specify two contrasting conditions, such as video- or text-based instructional feedback, and a particular homework topic, such as ``Adding Whole Numbers,'' or ``Factoring Quadratic Equations.'' Then, students working on that topic are individually randomized between the two conditions. The researcher could then compare the relative impact of video- vs.\ text-based feedback on an outcome variable of interest such as homework completion.  The anonymized data associated with the study, consisting of several levels of granularity and rich covariates describing both historical pre-study and within-study student interaction, is made available to the researcher.  The TestBed currently hosts over 100 such RCTs, and several of these RCTs have recently been analyzed, e.g., \citep{fyfe2016providing, walkington2019effect,prihar2022exploring,vanacore2023impact,gurung2023identification,gurung2023how}. 

In the ASSISTments TestBed example, a given RCT is likely to consist of just a few hundred students assigned to a specific homework assignment, limiting statistical power and precision. For instance, in one typical ASSISTments TestBed A/B test, a total of 294 students were randomized between two conditions, leading to a standard error of roughly four percentage points when estimating the effect on homework completion.  This standard error is too large to either determine the direction of a treatment effect or rule out clinically meaningful effect sizes. But the ASSISTments database contains data on hundreds of thousands of other  ASSISTments users, many of whom may have completed similar homework assignments, or who may have even completed an identical homework assignment but in a previous time period.  
 
This paper outlines an approach to estimate treatment effects in an RCT while incorporating high-dimensional covariate data, large observational remnant data, and machine learning prediction algorithms to improve precision. It does so without compromising the accuracy guarantees of traditional design-based RCT estimators, yielding unbiased point estimates and sampling variance estimates that are conservative in expectation; the approach is design-based, relying only on the randomization within the RCT to make these guarantees. In particular, the method prevents ``bias leakage'': bias that might have occurred due to differences between the remnant and the experimental sample, biased or incorrect modeling of covariates, or other data analysis flaws, does not leak into the RCT estimator. We combine recent causal methods for within-RCT covariate adjustment with other methods that have sought to incorporate high dimensional remnant data into RCT estimators.  In particular, we focus on the challenge of precisely estimating treatment effects from a set of 33 TestBed experiments \citep{data}, using prior log data from experimental participants and non-participants in the ASSISTments system.

The nexus of machine learning and causal inference has recently
experienced rapid and exciting development. This has included novel
methods to analyze observational studies,
e.g., \citep{diamond2013genetic}, to estimate subgroup effects,
e.g., \citep{
kunzel2018transfer}, or to optimally allocate treatment, e.g., \cite{uplift}.   Other developments share our goal, i.e., improving the precision of average treatment effect estimates from RCTs. These include the flexible approaches of 
\cite{aronowMiddleton,wager2016high,chernozhukov2018double}, 
all of which can incorporate arbitrary prediction methods, \cite{bloniarz2016lasso}, which uses the Lasso regression estimator to analyze experiments, and the Targeted Learning framework \citep{rosenblum2010simple,van2011targeted},
 which combines ensemble machine learning with semiparametric maximum likelihood estimation. 

A large literature has explored the possibility of improving precision in RCTs by pooling the controls in the RCT with  historical controls from observational datasets or from other similar RCTs.  This literature dates back at least to \cite{pocock1976combination}; for a review see \cite{viele2014use}. 
Much of this work uses a Bayesian framework, although frequentist approaches exist as well \citep{yuan2019design}.  In many of these methods biases can be arbitrary large depending on the choice of historical controls.  
Other recent efforts have sought to improve precision in RCT estimates by using the results of separate models fit on observational data.  
These include 
\cite{deng2013improving}, which fits a covariate model to pre-experimental data and then uses it to reduce standard errors of online A/B tests; 
\cite{gui2020combining}, which uses the RCT to de-bias a broken IV estimate obtained from observational data and then further combines this with an independent RCT-based estimate; and
\cite{opper2021improving}, which develops a variant of the sample-splitting estimator that we review below, and suggests a role for auxiliary data as well. 

Other literature has sought to combine effect estimates from experimental and observational studies, often under the framework of ``data fusion'' \citep{bareinboim2016causal}; these methods require observational data on both treated and untreated subjects. In addition to variance reduction, these methods may also seek to
generalize the results of RCTs to other populations or other outcome variables, improve the design of RCTs, detect problems in observational studies, or accomplish other goals
\citep{
hartman2015sate,
athey2020combining, 
rosenman2021designing,
rosenman2020combining,
rosenman2022propensity, 
chen2021minimax,
kallus2018removing}.
For recent reviews, see \cite{degtiar2021review, colnet2020causal}. 

A parallel literature in survey methodology discusses the possibility of combining probability and nonprobability samples in order to increase precision, especially for small area estimation
\citep{
 breidt2017model, 
erciulescu2020statistical,
dagdoug2021model, 
mcconville2020tutorial}.

In this paper, our goal is to estimate the average treatment effect within the RCT, and our focus is on using observational data---non-randomized subjects in the control or treatment conditions, or both, or neither---to improve the precision of the estimate. The main idea is to use observational data to train an algorithm that predicts RCT outcomes, and use the resulting predictions in the randomized sample as a new covariate. While this approach will work with any covariate adjustment technique, we suggest an approach based on the principal of ``first, do no harm,'' meaning that we prioritize retaining the advantages of randomized experiments highlighted above.  In particular, we seek to ensure that our method (1) does not introduce any bias, (2) will not harm precision, and ideally will improve precision, and (3) does not require any additional statistical assumptions beyond those typically made in design-based analysis of RCTs.

The paper is organized as follows.  Section \ref{sec:background} reviews background material, including design-based RCT analysis and covariate adjustment.  Section \ref{sec:theMethod} discusses incorporating remnant data, and presents our main methodological contribution.  In Section \ref{sec:assistments} we apply the method to estimate treatment effects in 33 TestBed experiments.  Section \ref{sec:conclusion} concludes.

\section{Methodological Background} \label{sec:background}
\subsection{Causal Inference from Experiments}\label{sec:causal-inference-from-experiments}

Consider a randomized experiment to estimate the average effect of a binary treatment $T$ on an outcome $Y$.  There are $N$ subjects, indexed by $i=1,\dots, N$.  Let $T_i = 1$ if subject $i$ is assigned to treatment, and $T_i = 0$ if control.  Let $\tg = \{i \mid T_i = 1\}$ and $\cg = \{i \mid T_i = 0\}$, and let $\nt = |\tg|$ and $\nc = |\cg|$.

Following \citet{neyman:1923} and \citet{rubin1974estimating}, let potential outcomes $\yti$ and $\yci$ represent the outcome value $Y_i$ that $i$ would have exhibited if he or she had (perhaps counterfactually) been assigned to treatment or control, respectively. We model the potential outcomes as fixed (not random).  \label{text:fixedpo}  Observed outcomes are a function of treatment assignment and potential outcomes:
\begin{equation*}
  Y_i=T_i\yti+(1-T_i)\yci
\end{equation*}
Define the treatment effect for $i$ as $\tau_i=\yti-\yci$. Our goal will be to estimate the average treatment effect (ATE), $\tbar \equiv \sum_i \tau_i / N = \ytbar - \ycbar$, where $\ytbar = \sum_{i=1}^N{\yti}/N$ is the mean of $\yt$ over all $N$ units in the experiment and $\ycbar$ is defined similarly.

If both $\yci $ and $\yti $ were known for each subject $i$, statistical modeling would be unnecessary---researchers could calculate $\tbar $ exactly, without error, by simply averaging observed $\tau$. In practice, we never observe both $\yci$ and $\yti$. Instead, we rely on the experimental setup to estimate and infer causation. Since the treatment and control groups are each random samples of the $N$ participants, survey sampling literature provides design-based unbiased estimators of $\ytbar$ and $\ycbar$ based on observed $Y$ and the known distribution of $T$. These estimators, and their associated inference, depend only on the experimental design, and not on modeling assumptions. The survey sample structure of randomized experiments allows us to infer counterfactual potential outcomes (at least on average) and estimate $\tbar$ as if $\tau_i$ were available for each $i$, albeit with sampling error.

We will use this framework to analyze the 33 TestBed experiments. These experiments are examples of ``Bernoulli experiments,'' in which each $T_i$ is an independent Bernoulli trial: $\Pr(T_i=1) = p$, with $0<p<1$, and $T_i\independent T_j \text{ if } i\ne j$. In the TestBed experiments, $p=1/2$. Estimation and inference about $\tbar$ is based on the observed values of $Y$ and $T$, and the known value of $p$. 

We will now introduce some statistical elements that we will use as the ingredients for our approach. Let $M_i=T_i\yci+(1-T_i)\yti$ denote $i$'s unobserved counterfactual outcome---when $i$ is treated, $M_i=\yci$ and when $i$ is in the control condition $M_i=\yti$. Then $i$'s treatment effect may be expressed as $\tau_i = (-1)^{T_i}(M_i - Y_i)$, i.e., $\tau_i=M_i-Y_i$ if $i$ is in the control group, or $\tau_i=Y_i-M_i$ if $i$ is in the treatment group. Although $M_i$ is, by definition, unobserved, it plays a central role in causal inference; its expectation,
\begin{equation*}
m_i\equiv \EE M_i=p\yci+(1-p)\yti
\end{equation*}
will also play a prominent role. Note that $m_i$ is a weighted average of subject $i$'s potential outcomes.

Let
\begin{equation*}
U_i=\begin{cases}
\frac{1}{p} & T_i=1\\
-\frac{1}{1-p} & T_i=0
\end{cases}
\end{equation*}
be subject $i$'s signed inverse probability weights; $U_i$ is merely a rescaled treatment indicator. Note that $\EE U_i=0$, and $\EE U_iY_i=\tau_i$. To see the latter, note that when $T=1$, with probability $p$, $Y_i=\yti$ and $U_iY_i=\yti/p$; when $T=0$, with probability $1-p$, $U_iY_i=-\yci/(1-p)$. Thus $U_iY_i$ may be thought of as an unbiased estimate of $\tau_i$, and $\thipw \equiv \sum_iU_iY_i/N$ is an unbiased estimate of $\bar{\tau}$. Note $\thipw$ is identical to the ``Horvitz-Thompson'' estimator of  \cite{aronowMiddleton}
\begin{equation}\label{eq:ipw}
\thipw=\frac{1}{N}\displaystyle\sum_{i\in\mathcal{T}}
\frac{Y_i}{p}-\frac{1}{N}\displaystyle\sum_{i\in\mathcal{C}}\frac{Y_i}{1-p}
\end{equation}
since it is the difference between the Horvitz-Thomson estimates of $\ytbar$ and $\ycbar$ \citep{horvitzThompson}.

The sampling variance of $\thipw$ proceeds from the same principals. The variance of $U_iY_i$ is 
\begin{equation}\label{eq:varTauHat1}
\var(U_iY_i)=\left(\yti\sqrt{\frac{1-p}{p}}+\yci\sqrt{\frac{p}{1-p}}\right)^2=\frac{m_i^2}{p(1-p)}
\end{equation}
and $\var(\thipw) = \sum_i m_i^2/[N^2 p(1-p)]$ because treatment assignments are independent. Note that because $\yti $ and $\yci $ are never simultaneously observed, $\var(\thipw)$ is not identified. However, $\hat{\var}(\thipw)=\sum_iU_i^2Y_i^2/N^2$ is an upper bound, i.e.,  $\EE\hat{\var}(\thipw)\ge \var(\thipw)$. (See \cite{aronowMiddleton} for equivalent expressions for more general experimental designs.)

Strangely, $\thipw$ and $\var(\thipw)$ are not translation-independent, i.e., adding a constant to each $Y$ changes both the value of $\thipw$ and $\var(\thipw)$ without changing the estimand $\bar{\tau}$. The more popular simple ``difference-in-means'' estimator \citep{neyman:1923},
\begin{equation}\label{eq:tauSD}
\tsd = \frac{1}{\nt}\sum_{i \in \tg} Y_i - \frac{1}{\nc}\sum_{i \in \cg} Y_i =\bar{Y}_{\tg}-\bar{Y}_{\cg}
\end{equation}
and its associated variance estimator 
\begin{equation} \label{eq:tsdvarhat}
\varhat(\tsd) = \frac{\syc}{\nc} + \frac{\syt}{\nt}
\end{equation}
where $\syc = \sum_{i \in\cg}(Y_i-\bar{Y}_{\cg})^2/(\nc-1)$ is the sample variance of the control group and $\syt$ is defined similarly, do not have this undesirable property. Our presentation here focuses on $\thipw$ as a jumping-off point for subsequent methodological development, but $\tsd$ will also play a prominent role.

\subsection{Design-Based Covariate Adjustment}\label{sec:cov-adj}

The reason for error when estimating $\tau$ is our inability to observe counterfactual potential outcomes $M$. As we have seen, randomized trials, coupled with design-based estimators like $\thipw$, use comparison groups and survey sampling theory to implicitly fill in this missing information. Baseline covariates---a vector $\bm{x}_i$ of data for subject $i$ gathered prior to treatment randomization---may potentially help us improve upon this strategy. Suppose a researcher has constructed algorithms $\predc$ and $\predt$ designed to impute $\yc$ and $\yt$, respectively, from $\bm{x}$. Then $\hat{M}_i = T_i \predcx + (1-T_i)\predtx$ is an imputation of $i$'s missing counterfactual outcome, and the researcher may estimate $\tau_i$ as $(-1)^{T_i}(\hat{M}_i-Y_i)$. In general, the bias of algorithms such as $\predc$ and $\predt$ will be unknown without further assumptions, so these effect estimates may be inadvisable. On the other hand, imperfect or potentially biased imputations of potential outcomes can, \emph{when combined with randomization}, yield substantial benefits.

The approach we will take to combining covariate adjustment with randomization has antecedents in 
\cite{robins1994estimation,scharfstein1999rejoinder,robins2000robust,rosenbaum:2002a,bang2005doubly,van2006targeted,tsiatis2008covariate,moore2009covariate,van2011targeted,aronowMiddleton,belloni2014inference,wager2016high,chernozhukov2018double,loop}, among others.  
We will focus on exactly unbiased estimators, despite the fact that a small amount of bias in finite sample is often acceptable, especially in the presence of other considerations. 
In fact, the covariate adjustment techniques we will develop have advantageous properties beyond unbiasedness (see, e.g. Section \ref{sec:ancova}).
That said, our main methodological contributions (in Section \ref{sec:theMethod}) are compatible with alternative techniques, including those that may be biased in finite samples.
We will frame our arguments around bias since we find it to be the easiest way to formalize confounding, which we see as the most pressing threat to estimators that include observational data.\label{text:unbiasedness}

In a Bernoulli experiment, note that
\begin{align*}
U_i(Y_i-m_i)
&=\begin{cases}
\frac{1}{p}(\yti-p\yci-(1-p)\yti)&T_i=1\\
-\frac{1}{1-p}(\yci-p\yci-(1-p)\yti)&T_i=0 \end{cases}\\
&=\begin{cases}
\frac{p(\yti-\yci)}{p}&T_i=1\\
\frac{(1-p)(\yti-\yci)}{1-p}&T_i=0
\end{cases}\\
&=\tau_i
\end{align*}
and this therefore suggests using imputations $\predcx$ and $\predtx$ to estimate $m_i$ as $\hmi = p\predcx + (1-p)\predtx$, and then estimating $\tau_i$ as
\begin{equation*}
\thim\equiv U_i(Y_i-\hmi).
\end{equation*}
For $\thim$ to be unbiased it is sufficient that algorithms $\predc$ and $\predt$ are constructed in such a way that
\begin{equation}\label{eq:indPred}
\{\predcx,\predtx\}\independent T_i.
\end{equation}
Since by design the distribution of $T_i$ does not depend on $\bx_i$,\label{text:x-independent} (\ref{eq:indPred}) is tantamount to requiring that $T_i$, and variables such as $Y_i$ that depend on $T_i$, play no role in constructing algorithms $\predc$ and $\predt$. Then, under (\ref{eq:indPred}), 
\begin{equation*}
\EE (\thim) = \EE (U_iY_i) - \EE (U_i\hmi) = \EE (U_iY_i) - \EE (U_i) \EE(\hmi) = \EE (U_iY_i) = \tau_i
\end{equation*}
where we use the facts that $\EE (U_i) = 0$ and $\EE (U_iY_i) = \tau_i$. Finally, define the ATE estimate:
\begin{equation} \label{eq:thm}
\thm=\frac{1}{N}\displaystyle\sum_{i=1}^N \thim = 
\frac{1}{N}\sum_{i \in \tg} \frac{Y_i-\hmi}{p} - 
\frac{1}{N}\sum_{i \in \cg} \frac{Y_i-\hmi}{1-p} 
\end{equation}
The unbiasedness of $\thm$ for $\bar{\tau}$ follows from the unbiasedness of each of its summands, $\thim$ for $\tau_i$.\label{text:biasExample}

Crucially, this unbiasedness holds even if $\predcx$ and $\predtx$ are biased; algorithms $\predc$ and $\predt$ need not be unbiased, consistent, or correct in any sense. As long as $\predcx$ and $\predtx$ are constructed to be independent of $T_i$, then $\thim$ will be unbiased. The same cannot be said for regression-based covariate adjustment, the common technique of regressing $Y$ on $T$ and $\bm{x}$ \citep{freedman2008regression}.

The estimate $\thm$ given in \eqref{eq:thm} is identical to the ``augmented IPW'' (AIPW) estimate familiar from the double-robustness literature in observational studies \citep[e.g.,][]{bang2005doubly}, but with known propensity scores $p$ \citep[see, e.g.][]{hahn1998role,rothe2016value}. Though AIPW estimators are typically derived in a model-based framework, the previous results show that in the context of an RCT, provided \eqref{eq:indPred} holds, the AIPW estimator \eqref{eq:thm} is unbiased under a design-based framework as well.\label{text:aipw} 

Compare $\thm$ to the estimate $\thipw$ given in \eqref{eq:ipw}. The only difference is that $Y_i$ in \eqref{eq:thm} has been replaced by $Y_i - \hmi$ in \eqref{eq:ipw}.  The goal of this covariate adjustment is to improve precision---we are residualizing our outcomes, in effect, to reduce variation.  Its success in this regard depends on the predictive accuracy of $\predcx$ and $\predtx$.  The variance of $\thim$ depends on $\hmi$ and is given by
\begin{equation}\label{eq:vthim}
\var(\thim \mid \hmi)=\frac{(\hmi-m_i)^2}{p(1-p)}.
\end{equation}
Compared with \eqref{eq:varTauHat1}, \eqref{eq:vthim} replaces $m_i$ with $\hmi-m_i$---that is, replaces potential outcomes with their residuals. Accurate imputations of $\yci$ and $\yti$, and hence of $\hmi$, yield precise estimation of $\tau_i$. On the other hand, inaccurate imputations, i.e., when $(\hmi-m_i)^2$ is greater than $m_i^2$, will decrease precision---though, again, without causing bias.
The sampling variance of the full estimator $\thm$ depends on how the parameters of $\predc$ and $\predt$ are estimated, which may induce dependence between $\hat{\tau}_i$ and $\hat{\tau}_j$ for $i\ne j$. 
The most important case, for our purposes, is discussed in the next section.\label{text:var-thm} 

\subsection{Sample Splitting}\label{sec:loop}

Successful covariate adjustment requires imputations $\predcx$ and $\predtx$ that are accurate and independent of $T_i$. To satisfy the independence condition, $i$'s observed outcome $Y_i$, which is a function of $T_i$, cannot play a role in the construction of the algorithms $\predc$ and $\predt$; they must be trained using other data.

This may be achieved by sample splitting, also referred to in this context as cross-estimation or cross-fitting.  In a Bernoulli experiment, rather than fitting global imputation algorithms $\predt$ and $\predc$ (which would violate \ref{eq:indPred}), fit a separate set of imputation models $\predti$ and $\predci$ for each experimental participant $i$, using data from the other participants.  In other words, for each $i$, one first drops observation $i$, and then use the remaining $N-1$ observations to construct imputation models for the control and treatment potential outcomes, denoted $\predci$ and $\predti$, respectively.  These models may be fit by any method, for example linear regression or random forests \citep{breiman2001random} (which, conveniently, automatically provides out-of-bag predictions for each subject).  In particular, methods that allow for regularization to prevent overfitting may be used. (For a discussion of sample-splitting for AIPW estimation, see, e.g. \cite{chernozhukov2018double,jiang2022new,smucler2019unifying}.)\label{text:cross-fit-aipw}

In this leave-one-out context,
\begin{equation*}
\hat{m}_i = p\predcxi + (1-p)\predtxi
\end{equation*}
and the estimated average treatment effect is then again given by $\tss = \sum_i \thim / N$ as in \eqref{eq:thm}, and where the superscript denotes ``sample splitting.''   Note that in a Bernoulli experiment $\hat{m}_i \independent T_i$ due to the fact that $\hat{m}_i$ is computed using $\bm{x}_i$ and a model fit without using observation $i$.  It follows that $\tss$ is unbiased.  Other randomization designs would call for modifications to the algorithm,  
e.g., \citep{ploop}.

When we wish to explicitly specify the covariates and imputation method that are used within $\tss$ we will write $\tss[\textrm{covariates}; \textrm{imputation method}]$.  For example, if we wished to use random forests and all available covariates we would write $\tss[\bx; \textrm{RF}]$, or if we wished to use only the fourth covariate and ordinary least squares regression we would write $\tss[x_4; \textrm{LS}]$.  If we wished to ignore the covariates and always set $\hat{m}_i = 0$ we would write $\tss[\varnothing; 0]$.  Note in particular that $\tss[\varnothing; 0] = \thipw$.

Building upon \eqref{eq:vthim}, and following \cite{loop}, the variance of $\tss$ may be estimated as follows.  Let
\begin{equation} \label{Echat}
\Ech = \frac{1}{\nc}\sum_{i \in \cg}\left[\predcxi - \yci\right]^2 
\end{equation}
be the mean-squared-error of control imputations $\predcx$ with respect to potential outcomes $\yc$, and define $\Eth$ similarly.  Note $\Ech$ and $\Eth$ are leave-one-out cross validation mean squared errors.  The estimated variance is then given by
\begin{align}
\hat{\var}(\tss) =  \frac{1}{N}\left[\frac{p}{1-p}\Ech + \frac{1-p}{p}\Eth  + 2\sqrt{\Ech \Eth }\right]. \label{loopvarhat}
\end{align}
This variance estimate will typically be somewhat conservative.  This is due to the fact that $\var(\tss)$ is unidentifiable, because the correlation of the potential outcomes is not estimable, and instead an upper bound is used \cite{loop}.  This difficulty is not unique to $\tss$; as noted in Section \ref{sec:causal-inference-from-experiments}, similar comments apply to $\thipw$, and the same is true of $\tsd$ as well \citep{neyman:1923,aronow2014}.  

Note that by \eqref{loopvarhat}, 
\begin{align} 
\hat{\var}(\tss)
&\le
\frac{\Ech }{N(1-p)} + \frac{\Eth }{Np} \nonumber \\
&\approx
\frac{\Ech }{\nc} + \frac{\Eth }{\nt} \label{t-test-like-variance}
\end{align}
which is similar in form to the variance estimate typically used in a two-sample $t$-test, namely $\frac{\syc}{\nc} + \frac{\syt}{\nt}$.  In (\ref{t-test-like-variance}), $\syc$ and $\syt$ are replaced by $\Ech $ and $\Eth $.  In other words, the sample variances are replaced by the estimated mean squared errors of the imputations.  

A special case occurs when the potential outcomes are imputed by simply taking the mean of the observed  outcomes (after dropping observation $i$).  That is, we set
\begin{equation}
\predcxi = \frac{1}{|{\cg \setminus i}|}\sum_{j \in \cg \setminus i} \ycj \label{meanimpute}
\end{equation}
and similarly for $\predtxi$.  Note that the covariates are simply ignored, and we denote this special case by $\tss[\varnothing; \textrm{mean}]$.  It can be shown that $\tss[\varnothing; \textrm{mean}] = \tsd$, i.e., the sample splitting estimator using leave-one-out mean imputation is exactly equal to the simple difference-in-means estimator.  Moreover, in this special case $\Ech = \frac{\nc}{\nc-1}\syc$ and $\Eth = \frac{\nt}{\nt-1}\syt$ and thus the variance estimate given by (\ref{t-test-like-variance}) is nearly identical to the ordinary $t$-test variance estimate \citep{loop}.

In short, when using mean imputation for the potential outcomes, the leave-one-out sample splitting procedure essentially simplifies to a standard $t$-test.  The effect estimate is identical, and the variance estimate is nearly identical.\footnote{These statements are conditional on $n_c \ge 2$ and $n_t \ge 2$.  When $n_c < 2$, then $\syc$ and the expression in (\ref{meanimpute}) are not defined.  When $n_c = 0$, $\bar{Y}_\mathcal{C}$ and $\tsd$ are also undefined.  More generally, several of our estimators are undefined when $n_c = 0$, namely $\tsd$ defined in (\ref{eq:tauSD}), $\Ech$ defined in (\ref{Echat}), as well as $\trebar$ in  (\ref{eq:taurebar}) and $\tgr$ in (\ref{eq:generalizedrebar}) defined in the next section.  Thus, it is worth noting that when we assert $\tsd$ is unbiased, we implicitly condition on $n_c, n_t > 0$.  (It is well known that $\tsd$ is unbiased conditional on any $n_c, n_t$, so long as $n_c, n_t > 0$. Without conditioning on $n_c$ and $n_t$, the moments of $\tsd$ are undefined in a Bernoulli trial.  See, e.g., \cite{freedman2007statistics}.)  The same applies to $\trebar$ and $\tgr$ in the next section.  For $\tss$ we do not implicitly condition $n_c, n_t > 0$ but rather assume that $\hat{m}_i$ is defined for all possible randomizations, including those in which $n_c < 2$ or $n_t < 2$.  This may be accomplished, for example, by setting $\hat{m}_i = 0$ in cases where $\predcxi$ or $\predtxi$ are otherwise undefined, in which case $\hat{\tau}_i$ reverts to the Horvitz-Thompson estimator.  As for $\Ech$ defined in (\ref{Echat}), we note that we could alternatively replace the $n_c$ in the denominator with $N(1-p)$, in which case $\Ech$ would be an unbiased estimate of $\frac{1}{N}\sum_{i=1}^N\textrm{MSE}[\predcxi]$.  In practice, we prefer to divide by $n_c$, although, unlike $\tsd$, we cannot claim that $\Ech$ is unbiased conditional $n_c, n_t > 0$.  See \cite{loop}.
}
This is highly reassuring.  Any imputation strategy that improves upon mean imputation in terms of mean squared error will reduce the variance of $\tss$ relative to $\tsd$.  Most modern machine learning methods employ some form of regularization to guard against overfitting, and thus typically perform no worse, or at least not substantially worse, than mean-imputation.  Thus in practice there is relatively little risk of hurting precision.\footnote{Beyond the question of \emph{hurting} precision, one might reasonably ask---as an anonymous reviewer did---whether, or in what sense, $\tss$ is optimal. 
Since $\tss$ is a version of the AIPW estimator, we may refer to the extensive literature on its optimality. For example, \cite{van2006targeted} gives a set of conditions under which AIPW is efficient or locally efficient, \cite{rothe2016value} discusses the case of a known propensity score, and \cite{chernozhukov2018double,jiang2022new}  discuss the sample-splitting AIPW estimator.
In general, the theoretical literature surrounding AIPW tends take potential outcomes as random, whereas in our development they are fixed; we defer an examination of the consequences of that distinction for future research.}

\section{Incorporating Observational Data} \label{sec:theMethod}
Modern field trials are often conducted within a very data-rich context, in which rich high-dimensional covariate data is automatically, or already, collected for all experiment participants. For instance, in the TestBed experiments, system administrators have access to log data for every problem and skill builder each participating student worked before the onset of the experiment. In other contexts, such as healthcare or education, rich administrative data is often available. In fact, these covariates are available for a much wider population than just the experimental participants---in the TestBed case, there is log data for all ASSISTments users. In other education or healthcare examples, administrative data is often available for every student or patient in the system, not just for those who were randomized to a treatment or control condition. Often, as in the TestBed case, the outcome variable $Y$ is also drawn from administrative or log data. We refer to subjects within the same data system in which the experiment took place---i.e. for whom covariate and outcome data are available---but who were not part of the experiment, as the ``remnant'' from the experiment. The remnant from a TestBed experiment consists of all ASSISTments users for whom log data is available but who did not participate in the experiment, of whom there are several hundred thousand. 

Simply pooling data from the remnant with data from the experiment undermines the randomization, since students in the remnant were not randomized between conditions.
This section will describe an alternative approach---a set of unbiased effect estimators that use the remnant to improve precision.
The estimators all begin by using the remnant to fit or train a model predicting potential outcomes as a function of covariates, and using that model to impute potential outcomes for units in the experiment.
They differ in how they use those imputations, and build on each other.
The following subsection discusses a simple residualizing estimator, Section \ref{sec:reloop} discusses sample splitting to improve that estimator, and Section \ref{sec:reloopplus} discusses incorporating an additional set of covariate-adjustment models fit to data from the experimental subjects themselves. 

We will focus on the case in which the treatment condition in the remnant is constant, irrelevant, or just unobserved.
For instance, in the TestBed dataset the RCTs typically test an experimental intervention against ``business as usual,''  and subjects in the remnant were all exposed to the control condition. 
Extension to cases in which $T$ is observed in the remnant is straightforward, and will be discussed briefly in Section \ref{sec:conclusion}.\label{text:noTinRemnant}

\subsection{Covariate Adjustment Using the Remnant}\label{sec:intro.remnant}

Design based covariate adjustment requires imputation models $\predc$ and $\predt$; \cite{aronowMiddleton} suggests training those models using ``auxiliary data'' such as the remnant. 
In the TestBed, there is no basis for separate imputation of $\yc$ and $\yt$; instead, we use data from the remnant to train an algorithm $\predrfun$ to predict (generic) outcomes as a function of covariates. 
In some cases $\predrfun$ may be interpreted as predicting control outcomes, but in other cases the interpretation may be more opaque.

Regardless of the interpretation, the logic of Section \ref{sec:cov-adj} would suggest using $\predrfun$ to construct the estimator $\thm$ \eqref{eq:thm}, by setting $\hmi=\predrxi$, where $\predrxi$, $i=1,\dots N$ denotes predictions obtained by applying $\predrfun$ to members of the RCT.\label{text:aipw-remnant}\footnote{This estimator was also suggested by an anonymous reviewer.} 
This estimator is equivalent to the IPW estimator $\thipw$ \eqref{eq:ipw}, but with observed outcomes $Y$ replaced by residuals $R_i \equiv Y_i-\predrxi$, that is, $\sum_i U_iR_i/N$.
Along similar lines, \cite{rebarEDM} proposes conditioning on $\nc$ and $\nt$ and using a difference in means estimator (also see \cite{deng2013improving} for a similar suggestion):
\begin{equation}\label{eq:taurebar}
\trebar = \frac{1}{\nt}\sum_{i \in \tg} R_i - \frac{1}{\nc}\sum_{i \in \cg} R_i =\bar{R}_{\tg}-\bar{R}_{\cg}
\end{equation}
In what follows we will refer specifically to \eqref{eq:taurebar} as ``the remnant estimator.''  

The remnant estimator $\trebar$ and its IPW variant work because $R_i$ is itself an outcome variable, with its own potential outcomes $\rci=\yci-\predrxi$ and $\rti=\yti-\predrxi$, and because $\predrxi$ is invariant to treatment assignment. Thus, treatment effects on the original outcomes are equal to treatment effects on the residualized outcomes, i.e.,
\begin{equation*}
\rti - \rci = \left[\yti - \predrxi \right] - \left[\yc-\predrxi \right] = \yti - \yci = \tau_i.
\end{equation*}
and $\trebar$---a difference-in-means estimate of this effect---is therefore an unbiased estimate of $\bar{\tau}$.  
This property holds regardless of whether $\predrfun$ itself is unbiased, consistent, or ``correct'' in any sense; indeed, as suggested above, it may not even be clear precisely what $\predrfun$ is estimating.

The goal of residualization is to improve precision.
Since $\trebar$ is a difference-in-means estimator, its sampling variance can be conservatively estimated in a similar way as $\tsd$ \eqref{eq:tsdvarhat}, but, again, with $R$ replacing $Y$: 
\begin{equation} \label{eq:trebarvarhat}
\varhat(\trebar) = \frac{\src}{\nc} + \frac{\srt}{\nt}
\end{equation}
Comparing this expression to $\varhat(\tsd)$ given in \eqref{eq:tsdvarhat}, we see that the residualized estimator will have a lower variance than $\tsd$ if 
$\src < \syc$ and $\srt < \syt$.  In other words, we wish for $\predrx$ to capture at least some of the variation in $Y$, so that $R$ is less variable than $Y$.  This will be achieved in practice when $\predrfun$ does indeed successfully predict outcomes in the RCT---or, more generally, when the sample covariances between $\hat{y}^r$ and $Y$ for subjects with $T=0$ and $T=1$, respectively, are sufficiently large. 

Importantly for practitioners, as long as only remnant data is used, $\predrfun$ may be trained and assessed in any way.  This process can be iterative, so that an analyst may train a candidate model, assess its performance (perhaps with $k-$fold cross-validation), modify the algorithm, and repeat until achieving suitable performance.  Any modeling approach may be taken, so long as no data from the RCT is used.  Post-selection inference, which would be a serious concern if model selection were done used the RCT data (especially when the dimension of $\bm{x}$ is large and the sample size is small), does not apply here. 

Unfortunately, in some cases (see, e.g., Section \ref{sec:assistments}) the remnant estimator may have greater sampling variance than the $\tsd$. This will be the case if $\predrfun$, trained in the remnant, extrapolates poorly to the experimental sample---for instance, if the distribution of $Y$ conditional on $\bm{x}$ differs substantially between the remnant and the RCT.  
To make matters worse, the performance of $\predrfun$ in the experimental sample---where it counts---may not be checked directly to select a best model, since when fitting $\predrfun$ outcomes from the RCT can not be touched.\footnote{However, one may use covariate data from the RCT to anticipate $\predrfun$'s performance; Appendix \ref{sec:rem-exp-diff} describes our (unfortunately unsuccessful) attempt to do so. Future research may result in improved methods.}\label{text:footnoteComparingRemRct}

Thus, residualizing with $\predcxi$---i.e., replacing $Y$ with $R$ in an unbiased estimator of $\bar{\tau}$---will result in an unbiased, design-based estimator that may be substantially more precise than $\tsd$, but may also be less precise.
In other words, covariate adjustment using the remnant in this way is potentially fruitful, but risky. \label{text:biasExample2}

\subsection{Flexibly Incorporating Remnant-Based Imputations}\label{sec:reloop}

Consider a ``generalized remnant estimator''

\begin{equation} \label{eq:generalizedrebar}
\tgr \equiv \frac{1}{\nt}\sum_{i \in \tg} \left[Y_i - b\predrxi\right] - \frac{1}{\nc}\sum_{i \in \cg} \left[Y_i - b\predrxi\right]
\end{equation}
where $b$ is some prespecified constant. Note that in the special case $b = 1$ this is the remnant estimator $\trebar$, and in the special case $b = 0$ it is the simple difference-in-means $\tsd$.  Thus, following the discussion above, when $\predrfun$ extrapolates well to the RCT, we wish to set $b=1$, and when $\predrfun$ extrapolates poorly to the RCT, we wish to set $b=0$.  More typically, an intermediate value for $b$ may be optimal.  

The challenge is that we do not know \textit{a priori} how well $\predrfun$ extrapolates to the RCT, and therefore do not know the optimal choice for $b$.  
We will use sample splitting to overcome that challenge. 
First define $\predr \equiv \predrx$.  
That is, we compute the remnant-based predictions of RCT outcomes as above (i.e., $\predrx$), but now regard these predictions simply as a covariate to be used within the sample splitting estimator (i.e., $\predr$).  
Then we construct a sample splitting estimator $\tss$ using the following imputation method:
\begin{equation}
    \begin{split}
\predcxri =& \aci + \bci \predri\\
\predtxri =&\ati + \bti \predri\label{CLSdef1} 
\end{split}
\end{equation}

where we obtain $\aci$, $\bci$, $\ati$, and $\bti$ by ordinary least squares, i.e., let
\begin{align}
(\aci, \bci) &= \argmin_{(a,b)} \sum_{j \in \mathcal{C} \setminus i}\left[Y_j - \left(a + b\predrj \right) \right]^2 \label{LSdef}
\end{align}
and similarly for $(\ati, \bti)$. We denote the resulting estimator $\trc$.

The estimator $\trc$ will typically be preferable to the remnant estimator $\trebar$ because, for each observation $i$, the remaining $N-1$ observations of the RCT help determine the best use of $\predr_i$ in constructing $\hat{m}_i$.  For example, suppose that the $\predr$ are highly accurate imputations of the $\yc$ in the RCT.  In this case, we might expect $\aci \approx 0$ and $\bci \approx 1$ so that $\predcxri \approx \predri$, or in other words, the remnant based predictions would ``pass through'' the linear regression largely unmodified, so that $\trc \approx \trebar$.  However, in contrast to the remnant estimator, poor imputations $\predr$ will not necessarily harm precision in $\trc$.  Consider the extreme case in which the $\predr$ are pure noise.  We would then expect $\aci \approx \bar{Y}_{\cg \setminus i}$ and $\bci \approx 0$ so that $\predcxri \approx \bar{Y}_{\cg \setminus i}$.  That is, we would revert approximately to mean-imputation, so that $\trc \approx \tsd$.  In other words, the role of $\predr$ may be tempered according to the prediction accuracy of $\predrfun$ in the RCT.  We might therefore expect $\trc$ to nearly always outperform, or at least perform no worse than, $\trebar$ and $\tsd$.  This intuition is formalized in the following proposition:

\begin{restatable}{prop}{propreloopbetter}\label{prop:reloopbetter}
Let  $(\yC{1},\yT{1},\predr_1),\dots,(\yC{N},\yT{N},\predr_N)$ be IID samples from a population in which $\yc$, $\yt$, and $\predr$ have finite fourth moments, and where $-1 < \mathrm{corr}(\yc, \predr) < 1$ and $-1 < \mathrm{corr}(\yt, \predr) < 1$.  Let $b$ be a fixed constant.  Let $\varhat[\tgr]$ denote the estimated variance of $\tgr$, defined analogously to \eqref{eq:tsdvarhat} and \eqref{eq:trebarvarhat}.  Let $\varhat\left\{\trc\right\}$ denote the estimated variance of $\trc$, defined as in \eqref{loopvarhat}. Then as $N \rightarrow \infty$, 
\begin{equation*}
\frac{\varhat\left\{\trc\right\}}{\varhat[\tgr]} \overset{p}{\to} \phi(b) \le 1
\end{equation*}
where $\phi(b)$ is some constant that depends on $b$.
\end{restatable}
\begin{proof} See Appendix \ref{sec:propositions}. \end{proof}

Notably, although this proposition is asymptotic in nature, we expect it to be relevant even in relatively small samples, given that $\trc$ effectively only requires estimating two more parameters than $\tgr$ (i.e., the slope coefficients $\bci$ and $\bti$).  
The ASSISTments experiments we analyze in Section \ref{sec:assistments} appear to generally support this intuition; $\trc$ nearly always outperforms $\tsd$.  Indeed, we see the greatest performance gain in the RCT with the smallest sample size.

Importantly, because the $\predr$ are used only as a covariate, they do not necessarily need to accurately impute the potential outcomes in the RCT; rather, it suffices that they are merely predictive.  If the RCT is systematically different from the remnant,  e.g., the potential outcomes in the RCT differ in scale from those in the remnant, the $\predr$ will still be useful as long as they are correlated with the experimental potential outcomes.  Indeed, counterintuitively, it is even possible for $\trc$ to achieve precision gains if the $\predr$ are \emph{anticorrelated} with outcomes in the RCT.

In any event, regardless of the properties of $\predrfun$ or quality of the data in the remnant, $\trc$ remains unbiased, and its associated variance estimator remains conservative, because it relies on $\tss$, which has both of those properties, and because $\predr$ is a covariate, and invariant to treatment assignment.

\subsection{Combining Remnant-Based and Within-RCT Covariate Adjustment}\label{sec:reloopplus}

The estimator $\trc$ effectively solves the remnant estimator's main deficiencies.  However, $\trc$ largely neglects the RCT covariates, except to the extent that $\predr$ depends on $\bx$ through $\predrfun$.  Neglecting the RCT covariate data may be suboptimal, especially when $\predrfun$ is poorly predictive of outcomes in the RCT, perhaps due to systematic differences between the RCT and the remnant.  Our goal in this section is to augment the strategy of the previous section, so that the RCT covariate data may be more fully exploited.

Define
\begin{equation}\label{eq:tilde-x}
\xti \equiv (x_{i1}, x_{i2}, ..., x_{ip}, \predri)
\end{equation}
or in other words, $\xti$ is $\bxi$ augmented with $\predr_i$.  We may now compute $\tss$ using the augmented set of covariates $\xt$ instead of $\bx$.  The hope is that by including $\predr$ we can exploit information in the remnant in much the same way that $\trc$ does, while simultaneously performing a more standard within-RCT covariate adjustment.  For example, we might use random forests and compute $\trcprf$.

In general, the precision of the estimator will depend on the performance of the imputation strategy, and in particular, its ability to integrate information from the remnant, via $\predr$, with information from other covariates $\bm{x}$.  On the one hand, $\predr$ is a function of the other covariates and thus, in at least some sense, does not contain any additional information. However, the function $\predrfun$  is fitted on the remnant, which may be much larger than the experimental sample, and thus $\predrfun$ may be a more accurate imputation function than what we would be able to obtain using the RCT data alone.  In this sense, $\predr$ does contain additional information, which can be exploited by the imputation method by heavily weighting $\predr$ over the other covariates. 

On the other hand, if the $\predr$ are highly accurate, using them as a covariate within a nonlinear model like a random forest may be statistically inefficient compared to a linear model, as in $\trc$. Therefore, it may not always be clear whether a highly flexible method such as $\trcprf$ will outperform $\trc$; it depends on the quality of the imputations $\predr$ as well as the predictive power of the covariates in the experimental sample.

This suggests imputing potential outcomes using a specialized ensemble learner \cite{opitz1999popular}: a weighted average of linear regression using just $\predri$, as in $\trc$, and random forests using $\tilde{\bm{x}}$, as in $\trcprf$.  More specifically, let $\predcxriLS$ be the least squares imputation defined in \eqref{CLSdef1} and \eqref{LSdef}, i.e., the imputation used within $\trc$; note in particular that $\predcxriLS$ ignores all of the entries of $\xti$ except $\predri$.  Let  $\predcxriRF$ denote the imputation from a random forest regression of $Y_{\cg \setminus i}$ on $\xt_{\cg \setminus i}$. We then define an ensemble imputation  
\begin{equation}\label{eq:optimalImputation}
\predcxriEN = \gci \predcxriLS + (1-\gci) \predcxriRF
\end{equation}
which is an interpolation between $\predcxriLS$ and $\predcxriRF$, where the interpolation parameter $\gci$ is such that $0 \le \gci \le 1$ and is given by
\begin{equation*}
\gci = \argmin_{\gamma \in [0,1]} \sum_{j \in \mathcal{C} \setminus i} \left\{Y_j - \left[\gamma \predcxrijLS + (1-\gamma) \predcxrijRF \right]\right\}^2
\end{equation*}
where $\predcxrijLS$ is defined analogously to $\predcxriLS$, but with both observations $i$ and $j$ removed, and similarly for $\predcxrijRF$.  That is, the interpolation parameter $\gci$ is obtained empirically to minimize mean squared error, and is obtained from a leave-one-out procedure, which ensures that $\gci \independent T_i$, and thus $\predcxriEN \independent T_i$.  We denote the resulting ensemble-based estimator $\trcpen$.  The imputation strategy \eqref{eq:optimalImputation} allows $\trcpen$ to triangulate between $\trc$ and $\trcprf$, and therefore combines the advantages of both, at the cost of estimating only one additional parameter (i.e., $\gamma_i^c$).

\section{Estimating Effects in 33 Online Experiments} \label{sec:assistments}
\subsection{Data from the ASSISTments TestBed}
\label{sec:assistmentsdata}

We apply and evaluate the methods described in this work to a set of 33 randomized controlled experiments run within the ASSISTments TestBed, described in the Introduction.  These A/B tests contrast a variety of pedagogical conditions in modules teaching 6th, 7th, and 8th grade mathematics content. 
For our purposes, the outcome of interest was completion of the module, a binary variable.

In general, once a TestBed proposal is approved, based on Institutional Review Board and content quality criteria, its experimental conditions are embedded into an ASSISTments assignment. This is then assigned to students, either by a group of teachers recruited by the researcher or, more commonly, by the existing population of teachers using ASSISTments in their classrooms. As an example, consider an experiment comparing text-based hints to video hints. The proposing researcher would create the alternative hints and embed them into particular assignable content, a ``problem set.'' Then, any time a teacher assigns that problem set to his or her students, those students are randomized to one of the conditions, and, when they request hints, receive them as either text or video.

There are several types of problem sets that researchers can utilize when developing their experiments. In the case of the 33 experiments observed in this work, the problem sets are mastery-based assignments called “skill builders.” As opposed to more traditional assignments requiring students to complete all problems assigned, skill builders require students to demonstrate a sufficient level of understanding in order to complete the assignment. By default, students must simply answer three consecutive problems correctly without the use of computer-provided aid such as hints or scaffolding (a type of aid that breaks the problem into smaller steps). In this way, completion acts as a measure of knowledge and understanding as well as persistence and learning, as students will be continuously given more problems until they are able to reach the completion threshold. ASSISTments also includes a “daily limit” of ten problems to encourage students to seek help if they are struggling to reach the threshold.

After the completion of a TestBed experiment, the proposing researcher may download a dataset which includes students' treatment assignments and their performance within the skill builder, including an indicator for completion. Additionally, the dataset includes aggregated features that describe student performance within the learning platform prior to random assignment for each respective experiment.  Summary statistics for the nine covariates we used in our analyses, pooled across experiments, are displayed in Table \ref{tab:covariates}.  These include the numbers of problems worked, and assignments and homework assigned, percent of problems correct on first try, assignments completed, and homework completed at the student and class level, and students' genders, as guessed by an internal ASSISTments algorithm based on first names.
We imputed missing covariate values separately within each experiment.  When possible, we used the mean of observed values from students in the same classroom; otherwise we used the grand mean.  We combined this data with disaggregated log data from students' individual prior assignments.

\begin{table}[ht]
\centering
\begin{tabular}{rrrr}
  \hline
 & Mean & SD & \% Missing \\ 
  \hline
Problem Count & 601.13 & 784.45 & 2 \\ 
  Percent Correct & 0.68 & 0.13 & 2 \\ 
  Assignments Assigned & 104.25 & 413.94 & 13 \\ 
  Percent Completion & 0.89 & 0.21 & 13 \\ 
  Class Percent Completion & 0.90 & 0.13 & 22 \\ 
  Homework Assigned & 25.97 & 29.90 & 50 \\ 
  Homework Percent Completion & 0.93 & 0.16 & 59 \\ 
  Class Homework Percent Completion & 0.93 & 0.09 & 56 \\ 
   Guessed Gender&Male: 36\%&Female: 36\%&Unknown: 28\%\\
 \hline
\end{tabular}
\caption{Summary statistics for aggregate prior ASSISTments performance used as within-sample covariates: number of problems worked, and assignments and homework assigned, percent of problems correct on first try, assignments completed, and homework completed at the student and class level, and students' genders, as guessed by ASSISTments based on first names.} 
\label{tab:covariates}
\end{table}

\subsection{Imputations from the Remnant}
\label{sec:33exp}
We also gathered analogous data from a large remnant of students who did not participate in any of the 33 experiments we analyzed. Ideally, the remnant would consist of previous ASSISTments students who had worked on the skill builders on which the 33 experiments had been run.  If that were the case, we would have considered 33 outcomes of interest, say $Y_{s}$, denoting completion of skill builder $s$.   Unfortunately, due to labeling conventions in the ASSISTments database, this was only feasible for 11 of the 33 experiments.  Instead, for all 33 experiments, we used prior ASSISTments data to impute one outcome, completion of a generic skill builder.

Rather than use the entire set of past ASSISTments users to build a remnant, we selected students who resembled those who participated in the 33 experiments. For the 11 experiments that we were able to match to other prior work, the remnant consisted of previous students who had worked on at least one of the skill builders in the experiments. For the remaining 22 experiments, we first observed the collection of problem sets given to students in the experiments before being assigned. The remnant consisted of all other ASSISTments users who had been assigned to at least one of those assignments. In other words, the remnant consisted of students who did not participate in any of the 33 experiments, but had worked on some of the same content as those who did. In all, the remnant consisted of 141,039 distinct students.  Sample sizes and skill builder completion rates in the 33 experiments are given in an online appendix in Table \ref{tab:info}.

We gathered records of up to ten assigned skill builders for each student in the remnant, and for each skill builder recorded the number of problems the student started, completed, requested help on, and answered correctly, the total amount of time spent, and assignment completion (i.e., skill mastery). Then, we fit a type of recurrent neural network \citep{williams1989learning} called Long-Short Term Memory (LSTM) \citep{hochreiter1997long} to the resulting panel data.  The model attempts to detect within-student trends in assignment completion and speed (i.e., the number of problems needed for skill mastery); please see Appendix \ref{sec:deepLearning} for further details. Using 10-fold cross validation within the remnant, we estimated the area under the ROC curve as 0.82 and a root mean squared error of 0.34 for the dependent measure of next assignment completion.

After fitting and validating the model in the remnant, we used it to predict skill builder completion for each subject in each of the 33 experiments.  To do so, we gathered log data for each student from up to ten previous assigned skill builders. (Students in the experiments with no prior data were dropped from all analyses.) Using the model fit in the remnant, we predicted whether each student would complete his or her next assigned skill builder. The resulting predictive probabilities were used as $\predr$ in the following analyses. 

\subsection{Results}

In each of the 33 experiments, we calculated five different unbiased ATE estimates: [1] the simple difference-in-means estimator $\tsd$ (equation \ref{eq:tauSD}); [2] the remnant estimator $\trebar$ (equation \ref{eq:taurebar}); [3] \ReLOOP\ (Section \ref{sec:reloop}); [4] \LOOP\ (Section \ref{sec:loop}) where $\bx$ denotes only those covariates supplied within the TestBed, as listed in Table \ref{tab:covariates}; and [5] \ReLOOPEN\ (Section \ref{sec:reloopplus}), using both $\predr$ and the provided TestBed covariates $\bx$.  These five methods are all design-based and unbiased, but they differ in their adjustment for covariates---both in the data they use for the adjustment, and in how the adjustment is effected.  Notably, in this application the remnant-based predictions $\predr$ are not functions only of $\bx$.  The covariates in $\bx$ are limited to aggregated data that summarize a student's previous performance (Table \ref{tab:covariates}), whereas the predictions $\predr$ are based on a more fine-grained longitudinal analysis of each student's log data.

Since each of these estimates is unbiased, we will focus on their estimated sampling variances. To aid interpretability, we will express contrasts between the sampling variances of two methods in terms of sample size. The estimated sampling variance of each estimator we consider is inversely proportional to sample size (see, e.g., equation \ref{loopvarhat}). Therefore, reducing the sampling variance of an estimator by, say, 1/2 is equivalent to doubling its sample size. Under that reasoning, the following discussion will refer to the ratio of estimated sampling variances as a ``sample size multiplier.''

\subsubsection{Remnant-Based Adjustment: Comparing $\trebar$ and \ReLOOP}

\begin{figure}
\centering
\input{images/fig4.tex}
\caption{A 
dotplot showing sample size multipliers (i.e.\ sampling
  variance ratios) comparing $\tsd$, $\trebar$, and \ReLOOP\ on the 33 ASSISTments TestBed experiments.}
\label{fig:ses1}
\end{figure}

Figure \ref{fig:ses1} compares $\tsd$, $\trebar$, and \ReLOOP\ on the 33 ASSISTments TestBed experiments. Each dot in the figure corresponds to a sample size multiplier comparing two estimated sampling variances in a particular experiment.  The vertical line at 1.0 indicates experiments in which the two methods gave approximately equal sampling variances. Dots to the right of the line correspond to experiments in which the variance in the denominator of the fraction was lower, and dots to the left of the line correspond to experiments in which the variance in the numerator was lower.

The leftmost plot contrasts $\trebar$ with $\tsd$. In four experiments, the variances of $\trebar$ and $\tsd$ were approximately equal, and in 27 experiments $\trebar$ outperformed $\tsd$. Notably, in one case (experiment \#33) the adjustment provided by $\trebar$ was equivalent to a roughly 85\% increase in sample size, and in another (experiment \#27) the adjustment was equivalent to a roughly 50\% increase. On the whole, $\trebar$ offers substantial gains in precision relative to $\tsd$. On the other hand, in two experiments the sampling variance of $\trebar$ was higher than that of $\tsd$.  Most notably, in one experiment (\#2) the adjustment given by $\trebar$ was equivalent to a roughly 45\% \emph{decrease} in sample size. In this case, apparently, the imputations from the model fit to the remnant were particularly inaccurate in the experimental sample. Because experimental outcomes played no role in determining the adjustment provided by $\trebar$, the adjustment was blind to this inaccuracy, and was unable to anticipate the resulting increase in variance in those cases.

In contrast, the \ReLOOP\ estimator incorporates information on imputation accuracy into its covariate adjustment. The middle panel of Figure \ref{fig:ses1} shows that across the board, \ReLOOP\ variances were smaller or roughly equal to those of $\tsd$. That is, \ReLOOP\ successfully avoided the risk that poor imputations pose to $\trebar$, and never increased variance relative to $\tsd$.  Moreover, in those cases in which $\trebar$ performed well, \ReLOOP\ tended to perform even better. For instance, in experiment \#33, the adjustment provided by \ReLOOP\ was equivalent to a roughly 100\% increase in sample size (relative to $\tsd$).

The rightmost panel of Figure \ref{fig:ses1} compares $\trebar$ to \ReLOOP\ explicitly: \ReLOOP\ sample variances dominated those of $\trebar$. In roughly half of the experiments, $\trebar$ and \ReLOOP\ performed similarly, and in the remaining half \ReLOOP\ improved upon $\trebar$.  Proposition \ref{prop:reloopbetter}, above, guarantees that \ReLOOP\ will dominate both $\tsd$ and $\trebar$ in the limit as $N\rightarrow\infty$; Figure \ref{fig:ses1} gives examples of this property in finite samples.

\subsubsection{Incorporating Standard Covariates}  

\begin{figure}
\centering
\input{images/fig5alt.tex}
\caption{A dotplot showing sample size multipliers (i.e.\ sampling
  variance ratios) comparing \ReLOOPEN\ to \ReLOOP\, \LOOP, and $\tsd$, respectively, on the 33 ASSISTments TestBed experiments.}
\label{fig:ses2}
\end{figure}

Figure \ref{fig:ses2} compares \ReLOOPEN\ to \ReLOOP, \LOOP, and $\tsd$, respectively, on the same 33 ASSISTments TestBed experiments. The left panel, comparing \ReLOOP\ to \ReLOOPEN, shows the impact of including standard covariates, incorporating them as described in the ensemble approach of Section \ref{sec:reloopplus}.  In all but one case, the sampling variance of \ReLOOPEN\ was less than or roughly equal to that of \ReLOOP---that is, including standard covariates improved precision.  In sixteen cases, this improvement was equivalent to increasing the sample size by more than 10\%; in eight of those cases the improvement was more than 25\% and in one experiment, it was more than 80\%. 

The middle panel compares the sampling variances of \ReLOOPEN\ and \LOOP, showing the extent to which including $\predr$ improved precision relative to using only standard covariates. 
In all but two experiments the sampling variance of \ReLOOPEN\ was less than or roughly equal to the sampling variance of \LOOP. In six experiments the improvement was equivalent to an increase in sample size of more than 10\%, and in one of those cases, experiment \#33, the improvement was equivalent to an over 65\% increase in sample size. 

The rightmost panel compares the sampling variances of \ReLOOPEN\ and the simple difference-in-means estimator, showing the total impact of covariate adjustment on statistical precision. Across every one of the 33 experiments, the estimated sampling variances for \ReLOOPEN\ were lower or roughly equal to those of $\tsd$. In 28 experiments the improvement was equivalent to increasing the sample size by more than 10\%; in 15 of those the improvement was equivalent to a more than 25\% increase in sample size, and in the case of experiment \#33, the improvement was equivalent to a 175\% increase in sample size.  

\FloatBarrier

\subsubsection{Covariate Adjustment with ANCOVA}\label{sec:ancova}

\begin{figure}
\centering
\input{images/OlsReloop.tex}
\caption{A dotplot showing sample size multipliers (i.e.\ sampling variance ratios), from contrasts between the  difference-in means estimator $\tsd$, sample-splitting estimators $\trc$ and $\trcpen$, and \textsc{ancova} estimators $\tls[\predr]$ and $\tls[\xt]$ with HC2 standard errors, on the 33 ASSISTments TestBed experiments.}
\label{fig:ols}
\end{figure}

The methodological development in Section \ref{sec:theMethod} focused on the covariate-adjusted estimator $\tss$, which can incorporate nearly any imputation method---including least squares regression, random forests, and ensemble methods such as \eqref{eq:optimalImputation}---while maintaining the advantages of design-based estimation, namely unbiased effect estimation and conservative standard error estimation. \label{text:ss-advantages}
However, the strategy of covariate adjustment using $\predr$ or $\xt$ is compatible with any covariate-adjusted estimator. 
For instance, an anonymous reviewer suggested estimating $\bar{\tau}$ via \textsc{ancova}---that is, fitting the model 
\begin{equation}\label{eq:ols}
Y_i=\mu+\beta T_i+\bm{\gamma}^T\bm{X}_i+\epsilon_i
\end{equation}
with ordinary least squares, where $\bm{X}_i=\predr_i$ or $\xt_i$ and estimating $\bar{\tau}$ with the estimated coefficient $\hat{\beta}$, which we will denote as $\tls[\predr]$ or $\tls[\xt]$, respectively (also see \cite{procova}). 
\textsc{Ancova} estimators $\tls[\cdot]$ are typically slightly biased, but consistent, with bias decreasing with $1/N$ \citep{freedman2008regression}.

Figure \ref{fig:ols} compares the estimated sampling variances of $\tsd$, $\trc$, $\trcpen$, $\tls[\predr]$ and $\tls[\xt]$ when applied to the 33 TestBed experiments.
(The \textsc{ancova} standard errors were estimated using the HC2 sandwich formula \citep[c.f.][]{mackinnon1985some}, the default for the \texttt{lm\_robust()} routine of the \texttt{estimatr} package in \texttt{R} \cite{estimatr,Rcite}.)
For the sake of comparison, the top panels of Figure \ref{fig:ols} reproduce results from Figures \ref{fig:ses1} and \ref{fig:ses2}, comparing $\tsd$ to $\trc$ and $\trcpen$.
The middle two panels contrast the sampling variances of $\tls[\predr]$ and $\tls[\xt]$ to $\tsd$.
Like $\trc$ and $\trcpen$, the \textsc{ancova} estimates are, in many cases, much more precise than $\tsd$. 
On the other hand, in some cases $\tls[\xt]$ had a noticeably higher sampling variance than $\tsd$---in one case, the effect of \textsc{ancova} adjustment was roughly equivalent to reducing the sample size by about 15\%.

Across the board, the precision gains afforded by $\tls[\xt]$ were typically slightly less than those afforded by $\trcpen$.
This is displayed in the bottom row of Figure \ref{fig:ols}, which compares the \textsc{ancova} estimators directly to $\trc$ and $\trcpen$. 
While $\trc$ and $\tls[\predr]$ tend to have very similar  sampling variances, $\trcpen$ is often (but not always) much more precise than $\tls[\xt]$.
Presumably, this advantage is due to the flexibility of the ensemble learner in $\trcpen$, which is in contrast to the linear additive adjustment of \textsc{ancova}.

\section{Discussion} \label{sec:conclusion}
 
Randomized experiments and observational studies have complementary strengths.  Randomized experiments allow for unbiased estimates with minimal statistical assumptions, but often suffer from small sample sizes.  Observational studies, by contrast, may offer huge sample sizes, but typically suffer from confounding biases which must be adjusted for, often through statistical modeling with questionable assumptions.  In this paper we have attempted to combine the strengths of both.  More specifically, we have sought to improve the precision of randomized experiments by exploiting the rich information available in a large observational dataset.

Our approach may be summarized as ``first, do no harm.''  A randomized experiment may be analyzed by taking a simple difference in means, which on its own provides a valid design-based unbiased estimate.  The rationale for a more complicated analysis would be to improve precision.  Our goal has therefore been to ensure that, in attempting to improve precision by incorporating observational data, we have not actually made matters worse.  In particular, we have sought to ensure that (1) no biases in the observational data may ``leak'' into the analysis, (2) we can reasonably expect to improve precision, not harm it, and (3) inference may be justified by the experimental randomization, without the need for additional statistical modeling assumptions.

In this paper, we focused on covariate adjustment using $\tss$, which is exactly unbiased; if a different covariate adjustment method were used instead of $\tss$, such as those proposed by \cite{lin2013agnostic} or \cite{guo2021generalized}, then the resulting estimator would inherit its properties, instead. 
We focus on the sample splitting estimator for two reasons. First, because we believe that a guarantee of exact unbiasedness will remove barriers to the method's adoption. Incorporating observational data into the analysis of RCTs may appear to be inherently risky, or to undermine the rationale for randomization. A general guarantee that effect estimates 
are unbiased, even in finite samples, may alleviate those concerns. 
Second, $\tss$ is compatible with nearly any imputation algorithm, 
and this flexibility may be especially valuable when incorporating $\predr$. 
The analysis in Section \ref{sec:ancova} provides a nice illustration of this: while there is little difference between the standard errors of $\trc$ and analogous \textsc{ancova} estimates, $\trcpen$---which uses an ensemble imputation algorithm including random forests---tended to perform substantially better than an \textsc{ancova} estimator using the same data. \label{text:unbiasedDiscussion}

The results from the 33 A/B tests we analyzed suggest that incorporating information gleaned from the remnant of an experiment can indeed improve causal inference---but it does not always do so.  The extent to which the remnant can help improve precision depends on the quality of the remnant-based predictions, and this in turn depends on both the quality of the remnant data and the algorithm $\predrfun$.  
It is therefore important to include observational data judiciously---our methods dynamically adapt, taking advantage of observational data when it is useful and minimizing its role when it isn't. 

The focus of this paper was to show that these methods can improve statistical precision without incurring a statistical cost---i.e. without potentially increasing bias or standard errors. However, gathering remnant data and using it to train an algorithm may require substantial human and/or computational resources. Therefore, it is crucial for applied researchers to be able to anticipate in advance the extent to which our methods will outperform estimators that use only RCT data. These cost benefit calculations can take place at two different points in the research process: before collecting any remnant data, and after collecting data from the remnant but before using it to train a predictive algorithm. Before collecting data from the remnant, researchers may be able to use observed properties of RCT data, along with anticipated, but yet unobserved, properties of the remnant to decide whether to proceed. For instance, some initial empirical results, currently under review, suggest that our methods have the potential to improve statistical precision across a wide range of RCT sample sizes, but that the most dramatic improvements tend to occur when the RCT sample size is small or moderate. Intuition suggests that the greatest contribution of auxiliary data will occur when a large number of covariates are available but there is little prior information on which covariates are the most important. If remnant data are available, analysts may decide whether to use it to train a predictive algorithm based on explicit comparisons between covariate distributions in the remnant and in the RCT (for example Appendix \ref{sec:rem-exp-diff}). Intuition suggests that our methods hold the greatest promise when covariates in the remnant and RCT are most similar.    

These, and other questions will be best answered by applying our methods in a wide variety of contexts. While we have focused on the ASSISTments platform in this paper, future work will explore what other sources of auxiliary data, and corresponding prediction algorithms, may be particularly well suited to improving the precision of RCTs typically encountered in education research.  Indeed, one of the advantages of developing models on observational data in this manner is that a wide variety of models may be explored, tested, and iteratively improved upon before they are applied to an RCT.

\label{text:both-conditions-in-remnant}
In particular, it will be interesting to consider cases in which the experimental condition varies---and is recorded---in the remnant.
For instance, the remnant from an RCT contrasting two common medical procedures may include medical records from previous patients who underwent one or the other procedure.
In that case, analysts may train remnant models to impute both potential outcomes as, say, $\hat{y}^{rc}(\bx)$ and $\hat{y}^{rt}(\bx)$. 
Then (following Section \ref{sec:intro.remnant}) they may set $\hmi=p\hat{y}^{rc}(\bx_i)+(1-p)\hat{y}^{rt}(\bx_i)$ and estimate average treatment effects using $\thm$, or (following Sections \ref{sec:reloop}--\ref{sec:reloopplus}), include $\hat{y}^{rc}(\bx)$ and $\hat{y}^{rt}(\bx)$ within a sample-splitting estimator $\tss$, perhaps alongside other covariates. 
We expect that the inclusion of both types of exposures in the remnant may enhance remnant-based estimators even further, and hope to explore these possibilities in future research.


\section{Acknowledgements}
We would like to thank Ben Hansen and Charlotte Mann for helpful discussions.  We would also like to thank the two anonymous reviewers for their comments.

\section{Funding information}
The research reported here was supported by the Institute of Education Sciences, U.S. Department of Education, through Grant R305D210031. The opinions expressed are those of the authors and do not represent views of the Institute or the U.S. Department of Education.  E.\ Wu was supported by NSF RTG grant DMS-1646108. N.\ Heffernan oversaw the creation of the 33 experiments and provided the data from ASSISTments; we want to acknowledge the funding that created/related to ASSISTments from 1) NSF (e.g., 2118725, 2118904, 1950683, 1917808, 1931523, 1940236, 1917713, 1903304, 1822830, 1759229, 1724889, 1636782, \& 1535428), 2) IES (e.g., R305N210049, R305D210031, R305A170137, R305A170243, R305A180401, R305D210036, R305A120125, \& R305R220012), 3) GAANN (e.g., P200A180088 \& P200A150306), 4) EIR (U411B190024 \& S411B210024), 5) ONR (N00014-18-1-2768), and 6) Schmidt Futures. None of the opinions expressed here are those of the funders.

\section{Conflict of Interest}
Authors state no conflict of interest.

\section{Code and Data}
Code and data are available at \url{https://osf.io/d9ujq/}

\bibliographystyle{unsrtnat}
\bibliography{rebarloop,rebarloop2}

\newpage
\appendix

\clearpage
\pagenumbering{arabic}
\setcounter{page}{1}

\clearpage
\section{Summary of A/B Test Data}

Table \ref{tab:info} gives sample sizes and skill builder completion rates in the 33 experiments discussed in the paper.

\begin{table}[!h]
\centering
\begin{tabular}[t]{rrrrr|rrrrr}
\hline
\multicolumn{1}{c}{ } & \multicolumn{2}{c}{n} & \multicolumn{2}{c}{\% Complete} & \multicolumn{1}{c}{ } & \multicolumn{2}{c}{n} & \multicolumn{2}{c}{\% Complete} \\
\cmidrule(l{3pt}r{3pt}){2-3} \cmidrule(l{3pt}r{3pt}){4-5} \cmidrule(l{3pt}r{3pt}){7-8} \cmidrule(l{3pt}r{3pt}){9-10}
Experiment & Trt & Ctl & Trt & Ctl & Experiment & Trt & Ctl & Trt & Ctl\\
\midrule
1 & 956 & 961 & 94 & 93 & 18 & 165 & 170 & 92 & 89\\
2 & 329 & 363 & 98 & 96 & 19 & 259 & 246 & 82 & 85\\
3 & 649 & 610 & 86 & 88 & 20 & 199 & 213 & 85 & 88\\
4 & 201 & 228 & 97 & 95 & 21 & 258 & 276 & 82 & 80\\
5 & 910 & 887 & 73 & 72 & 22 & 188 & 193 & 89 & 85\\
6 & 931 & 900 & 61 & 64 & 23 & 242 & 266 & 81 & 76\\
7 & 360 & 344 & 88 & 88 & 24 & 279 & 235 & 72 & 69\\
8 & 492 & 463 & 79 & 81 & 25 & 269 & 288 & 65 & 59\\
9 & 215 & 211 & 93 & 92 & 26 & 225 & 232 & 73 & 74\\
10 & 231 & 197 & 92 & 91 & 27 & 267 & 256 & 63 & 62\\
11 & 607 & 578 & 68 & 63 & 28 & 228 & 244 & 68 & 64\\
12 & 370 & 384 & 83 & 82 & 29 & 239 & 258 & 54 & 48\\
13 & 338 & 289 & 88 & 84 & 30 & 74 & 92 & 91 & 84\\
14 & 478 & 476 & 76 & 73 & 31 & 69 & 67 & 91 & 87\\
15 & 193 & 209 & 89 & 93 & 32 & 76 & 81 & 62 & 70\\
16 & 404 & 451 & 73 & 69 & 33 & 15 & 11 & 73 & 55\\
17 & 264 & 274 & 84 & 85 &  &  &  &  & \\
\bottomrule
\end{tabular}
\caption{\label{tab:info}Sample sizes and \% homework completion---the outcome of interest---by treatment group in each of the 33 A/B tests.}
\end{table}

\FloatBarrier

\section{Proof of Proposition} \label{sec:propositions}

\propreloopbetter*
\begin{proof}
We first explicitly define $\varhat[\tgr]$.  Let $\rgr_i = Y_i - b\predri$ and define
\begin{equation} \label{varhattgr}
    \varhat[\tgr] = \frac{\srgrc}{\nc} + \frac{\srgrt}{\nt}.
\end{equation}
Comparing (\ref{varhattgr}) to (\ref{t-test-like-variance}) we see that in order to prove the desired result, it is sufficient to show that $\Ech/\srgrc \pto \phi_c(b) \le 1$ and $\Eth/\srgrt \pto \phi_t(b) \le 1$ where $\phi_c(b)$ and $\phi_t(b)$ are constants that depend on $b$.  We will show $\Ech/\srgrc \pto \phi_c(b) \le 1$; the argument for $\Eth/\srgrt \pto \phi_t(b) \le 1$ is analogous.

Let $\Ect$ be defined similarly to $\Ech$, except that $\Ect$ does not use leave-one-out predictions, and instead uses predictions based on all of the data.  That is,
\begin{equation} \label{Ectilde}
\Ect = \frac{1}{\nc}\sum_{i \in \cg}\left[\predcxrit - \yci\right]^2 
\end{equation}
where $\predcxrit = \tilde{a}^c + \tilde{b}^c\predri$ and where $\tilde{a}^c$ and $\tilde{b}^c$ are the intercept and slope coefficients, respectively, from a univariate regression of $Y_{\cg}$ on $\predr_{\cg}$ (not dropping observation $i$).  Now note the following: (a) both $\srgrc$ and $\Ect$ converge to finite constants; (b) the constant to which $\srgrc$ converges is not 0; and (c) $\Ect \le \srgrc$ for all $n_c \ge 2$.  (a) is ensured by the moment conditions.  (b) is ensured by the condition $-1 < \mathrm{corr}(\yc, \predr) < 1$.  (c) follows from the fact that $\Ect = \frac{1}{\nc} \min_{(a,b)} \sum_{j \in \mathcal{C} }\left[Y_j - \left(a + b\predrj \right) \right]^2$ whereas $\srgrc = \frac{1}{\nc-1} \sum_{j \in \mathcal{C} }\left[Y_j - b\predrj -\overline{Y_j-b\predrj} \right]^2=\frac{1}{\nc-1} \min_{a} \sum_{j \in \mathcal{C} }\left[Y_j - \left(a + b\predrj \right) \right]^2$ for a fixed value of $b$, and thus the minimization problem of the former is less constrained than the latter.  As a result of (a), (b), and (c), it follows that $\Ect/\srgrc \pto \tilde{\phi}_c(b) \le 1$.

To complete the proof, it suffices to show that $\Ech \pto \Ect$.  After some algebra,
\begin{equation}
    \Ech = \frac{1}{\nc}\sum_{i \in \cg}\left[\predcxrit - \yci\right]^2/(1-h_i)^2 
\end{equation}
where 
\begin{equation}
h_i = \frac{1}{(\nc - 1)S^2(\predr_\cg)}\left[\overline{(\predr_\cg)^2} -2\overline{\predr_\cg}\predri + (\predri)^2\right].
\end{equation}
Here, the $h_i$ are the diagonal entries of the hat matrix from the regression of $Y_{\cg}$ on $\predr_{\cg}$ and we use the well-known shortcut formula for calculating leave-one-out residuals \cite{seber2012linear}.  Note that $0 < h_i \le 1$.  Thus,
\begin{equation}
    |\Ech - \Ect| \le \left[\frac{1}{(1-h^*)^2}-1\right] \Ect  
\end{equation}
where $h^* = \max_\cg h_i$.  However, because of the moment conditions on $\predr$, it is straightforward to show that $h^* \pto 0$, and therefore $\Ech \pto \Ect$.
\end{proof}

\newpage

\section{Deep Learning in the Remnant to Impute Completion}\label{sec:deepLearning}
\noindent

We used the remnant to train a variant of a recurrent neural network \citep{williams1989learning} called a Long-Short Term Memory (LSTM) network \citep{hochreiter1997long} to predict students' assignment completion. Deep learning models, and particularly LSTM networks, have been previously applied successfully to model similar temporal relationships in various areas of educational research \citep{piech2015deep,botelho2017improving}. 

Neural networks, including recurrent networks such as those explored here, are universal function approximators \citep{hornik1989multilayer,schafer2006recurrent}. These models are commonly represented as ``layers'' of neurons; these feed from a set of inputs, through one or more ``hidden'' layers, to an output layer, where, in the basic case, the output of each layer is determined by Equation~\ref{eq:ffnn}. In that equation, $W$ is a set of learned weights, comparable to the coefficients learned in a regression model. The activation function $a(.)$ is commonly a non-linearity that is applied to each layer in the network.
\begin{equation}
\label{eq:ffnn}
    h_\ell = a( W * h_{\ell-1} + b) 
\end{equation}
where $h_0$ is the input vector $X$.

Recurrent networks build upon this formulation to add layers that utilize not only the outputs of preceding layers, but also incorporate values from previous time steps within a supplied series; in time series data, the model estimates for a particular time step may be better informed by information from previous time steps, and a recurrent network structure is designed to take advantage of this likelihood. The LSTM networks explored here incorporate a set of ``gates'' that regulate the flow of data from both preceding layers and a ``cell memory'' that is calculated through previous time steps. The output of this LSTM layer is given by Equations~\ref{eq:lstmf}-\ref{eq:lstmh}. 
\begin{equation}
    \label{eq:lstmf}
    f_t = \sigma( W_f * [h_{t-1}, x_t] + b_f)
\end{equation}
\begin{equation}
    \label{eq:lstmi}
    i_t = \sigma( W_i * [h_{t-1}, x_t] + b_i)
\end{equation}
\begin{equation}
    \label{eq:lstmo}
    o_t = \sigma( W_o * [h_{t-1}, x_t] + b_o)
\end{equation}

\begin{equation}
    \label{eq:lstmtc}
   \Tilde{C}_t = \tanh( W_C * [h_{t-1}, x_t] + b_C)
\end{equation}
\begin{equation}
    \label{eq:lstmc}
   C_t = f_t*C_{t-1} + i_t * \Tilde{C}_t
\end{equation}

\begin{equation}
    \label{eq:lstmh}
   h_t = o_t*\tanh(C_t) 
\end{equation}
where $t$ is given as recurrent layer $\ell$ on the given timestep.

In the above equations, gates $f_t$ and $i_t$ inform the cell memory $C_t$ how much of the previously-computed memory should be forgotten and updated with the output of the previous time step and preceding layer, respectively.

As a recurrent network, the model is trained by iteratively updating the weight matrices ($W$ in the above equations) through a procedure known as backpropagation through time  \citep{werbos1990backprop} combined with a stochastic gradient descent method called Adam \citep{kingma2014adam}. These methods are informed by a cost function (sometimes called a loss function) that is calculated through the comparison of model predictions with supplied ground truth labels. In this work, we adopted a network structure that incorporates multi-task learning \citep{caruana1997multitask} as a means of regularization. In other words, our model ultimately produces two sets of predictions corresponding with two outcomes of interest: student completion and inverse mastery speed, each on the subsequent assignment. By optimizing model weights in regard to these two outcomes, the process helps prevent the model from overfitting to either outcome; as student completion of their next assignment is the outcome explored in this work, the second outcome of inverse mastery speed is used only for this regularization purpose and is not utilized in subsequent analyses. Given that student completion is binary and inverse mastery speed is a continuous measure, the formula of which is described in Table~\ref{tab:modelfts}, the cost function for our model training was calculated as a linear combination of two separate cost functions. Binary cross-entropy is used in the case of next assignment completion, as shown in Equation~\ref{eq:BCE}, while RMSE (Equation~\ref{eq:RMSE}) is used in the case of inverse mastery speed on the next assignment. The final cost function is then given as Equation~\ref{eq:cost}, which is calculated over smaller smaller ``batches'' of samples over multiple training cycles known as epochs.
\begin{equation}
\label{eq:BCE}
    \mathrm{BCE} = -(y*\log(\hat{y})+(1-y)*\log(1-\hat{y}))
\end{equation}
\begin{equation}
\label{eq:RMSE}
    \mathrm{RMSE} =  \sqrt{\frac{1}{n}\sum(y-\hat{y})^2}
\end{equation}
\begin{equation}
\label{eq:cost}
    \mathrm{Cost}_{\mathrm{batch}} = \frac{\mathrm{BCE}_{\mathrm{batch}}+\mathrm{RMSE}_{\mathrm{batch}}}{2}
\end{equation}
The training of the model continues by calculating the cost and iteratively updating model weights over multiple epochs until a stopping criterion is met. In this regard, we hold out 30\% of the training data as a validation set. Model performance is calculated on this validation set after each epoch of training. Training ceases once the model performance on this validation set stops improving (i.e., the difference of model performance from one epoch to the next falls below a designed threshold). To avoid stopping the training process too early due to small fluctuations in model performance on the validation set early in the training procedure, a 5-epoch moving average of validation cost is used as the stopping criterion.

The specific model structure used in this work observed an LSTM network comprised of 3 layers. We used 16 covariates to describe each single time step, which then feeds into a hidden LSTM layer of 100 nodes, which is used to inform an output layer of two units corresponding with the previously described two outcomes of interest. The input features used in this model, described in Table~\ref{tab:modelfts}, represent transformed and non-transformed versions of several metrics that describe different aspects of student performance within a single assignment. We considered sequences of at most ten worked skill builder assignments (c.f. Section~\ref{sec:assistmentsdata}), to predict student completion on a subsequent skill builder assignment.

\begin{table}[ht!]
\centering
\begin{tabular}{p{1.5in}p{3.5in}}
\hline
\textbf{Input Feature} & \textbf{Description} \\ \hline
Problems Started &
  The number of problems started by the student. (Untransformed \& Sq.Root) \vspace{4pt} \\ 
Problems Completed &
  The number of problems completed by the  student. (Untransformed \& Sq.Root) \vspace{4pt} \\ 
Inverse Mastery Speed &
  The inverse of the number of problems needed to complete the mastery assignment, or 0 where the student did not complete. (Untransformed \&  Sq.Root) \vspace{4pt} \\ 
Percent Correct &
  The percentage of problems answered correctly on the first attempt without the use of hints. (Untransformed \& Sq.Root) \vspace{4pt} \\ 
Assignment Completion &
  Whether the current assignment was completed by the student. \vspace{4pt} \\ 
Attempts Per Problem &
  The number of attempts taken to correctly answer each problem. (Avg. \& Sq.Root) \vspace{4pt} \\ 
First Response Time &
  The time taken per problem before making the first action. (Avg.) \vspace{4pt} \\ 
Problem Duration &
  The time, in seconds, needed to solve each problem. (Avg.) \vspace{4pt} \\ 
Days with Activity &
  The number of distinct days on which the student worked on each problem in the assignment. (Avg.) \vspace{4pt} \\ 
Attempted Problem First &
  Whether, on each problem, the first action was an attempt to answer (as opposed to a help request). (Avg.) \vspace{4pt} \\ 
Requested Answer Hint &
  Whether, on each problem, the student needed to be given the answer to progress. (Avg.) \vspace{4pt} \\ \hline
\end{tabular}
\caption{Assignment-level features in LSTM Model}
\label{tab:modelfts}
\end{table}

We specified the LSTM model's hyperparameters (e.g., number of LSTM nodes, delta of stopping criterion, weight update step size) based on previously successful model structures and training procedures within the context of education. We evaluated the model using a 10-fold cross validation within the remnant to gain a measure of model fit (leading to an ROC area under the curve of 0.82 and root mean squared error of 0.34 for the dependent measure of next assignment completion). After this evaluation, the model is then re-trained using the full set of remnant data. This trained model is then used within the analyses described in Section~\ref{sec:assistments}.

\section{Comparing Covariates in the Remnant to the RCT}\label{sec:rem-exp-diff}
The requirement \eqref{eq:indPred} that imputations $\predcx$ and $\predtx$ are independent of treatment assignment $T_i$ precludes any use of RCT outcomes in training the imputation algorithm $\predrfun$.
This is due to the fact that, if there is indeed a treatment effect for any RCT subject, RCT outcomes are a function of $T$.
This restriction includes the use of $Y$ to select between competing $\predrfun$ algorithms, or to decide whether to use remnant-based predictions $\hat{y}^r$ for covariate adjustment at all.
That is, so long as analysts use only remnant outcomes, they may assess and modify $\predrfun$ without restriction without violating \eqref{eq:indPred}, they may not use outcome data from the RCT.

This restriction, however, does not extend to covariate data $\bm{x}$ from the RCT. An anonymous reviewer suggested developing a method comparing covariate distributions between the RCT and the remnant that may indicate the gain in precision an analyst may expect from including $\hat{y}^r$ in a covariate adjustment estimator.

Here we discuss a technique we attempted, although we do not believe that it achieved its aim.

The intuition behind our approach is based roughly on ``K-Nearest Neighbors'' classification---if a subject in the RCT closely resembles other subjects in the remnant, an algorithm trained on the remnant may be able to predict that subject's outcome accurately, whereas if an RCT subject is unlike many other remnant subjects, the prediction is not likely to be accurate.
Formally, let $K>0$ be an integer, and $D(\cdot,\cdot)$ be a distance measure.
Then, for subject $i$ in the RCT and subject $j$ in the remnant, let $d_{ij}=D(\bx_i,\bx_j)$, then, for each $i$, sort these distances so that $d_{i(1)}\le d_{i(2)}\le \dots $.
Finally, compute $\bar{d}_i^k=\sum_{k=1}^K d_{i(k)}/K$, the average distance between $\bm{x}_i$ and it's $K$ nearest neighbors.
The thought is that outcomes for subjects with low $\bar{d}_i^K$ should typically be easier to predict than subjects with larger $\bar{d}_i^K$.
Distances within the remnant may form a reasonable basis of comparison---that is, one may compare $\bar{d}_i^k$ to the distribution of average distances between remnant subjects and their $K$ nearest neighbors.

\begin{figure}
    \centering
    \includegraphics{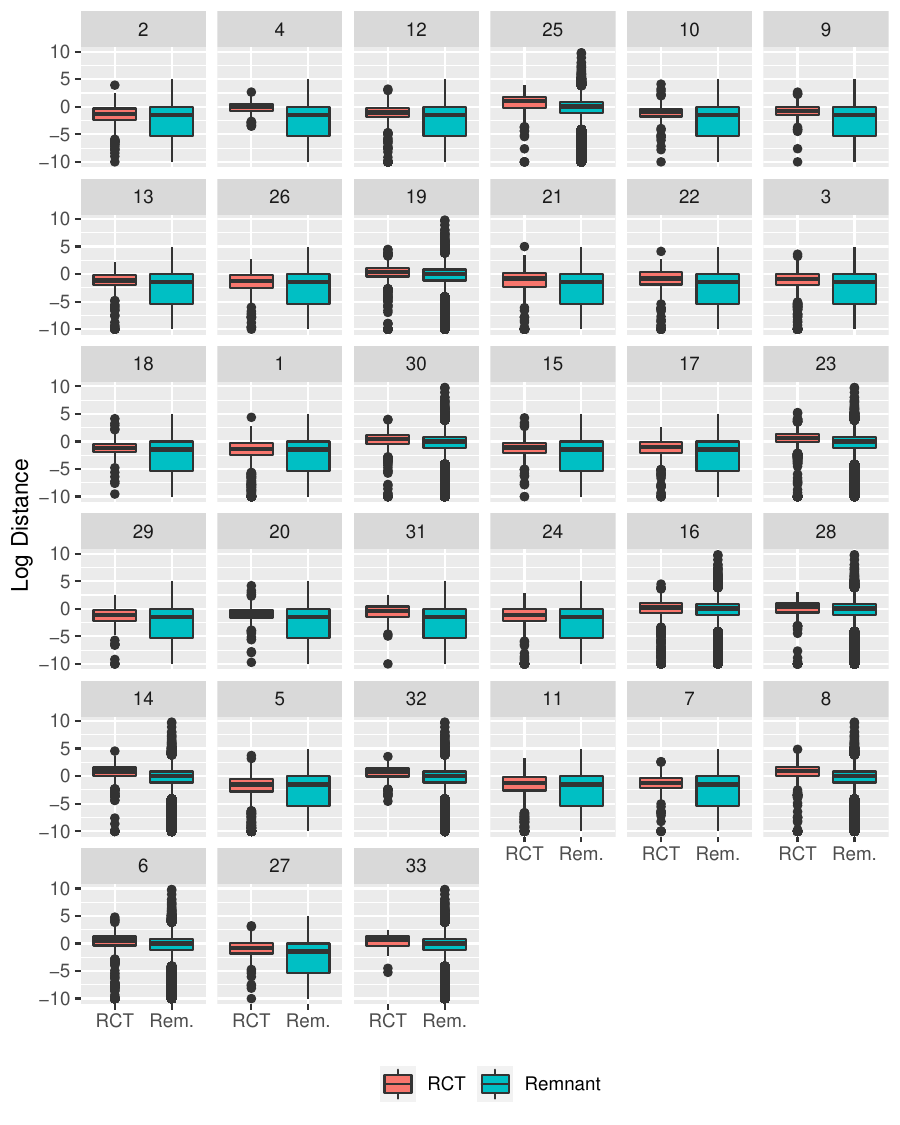}
    \caption{Boxplots comparing the distribution of $\bar{d}_i^5$ for each ASSISTments TestBed A/B test against the analogous distribution for the corresponding remnant. Panels are ordered lowest to highest according to $\varhat(\tsd)/\varhat(\trc)$}
    \label{fig:rem-exp-diff}
\end{figure}

To calculate this measure for TestBed A/B tests, we first flattened each subject's covariate data by averaging their assignment-level statistics and also including a covariate equal to the number of included assignments.
Then, we chose $K=5$ and $D(\cdot,\cdot)$ to be the Mahalanobis distance, with the covariance matrix estimated using the remnant.

Figure \ref{fig:rem-exp-diff} shows the results.
Each panel corresponds to a different A/B test, and displays a boxplot of $\bar{d}_i^K$ for subjects in the RCT next to an analogous boxplot for the remnant.
Note that across the 33 experiments, there were two distinct remnants, corresponding to two separate data draws.
The panels are sorted lowest to highest according to the ratio $\varhat(\tsd)/\varhat(\trc)$.

Unfortunately, no pattern is apparent, suggesting that $d_i^5$ is not a useful indicator of the variance reduction potential of algorithms trained in the remnant.
Future research may lead to modifications of $d_i^K$ or another measure entirely that may better anticipate $\predrfun$'s out-of-sample performance.
Fortunately, estimators $\trc$ and $\trcpen$ often perform well, and (in our examples) never harm precision, even when $\predrfun$ performs poorly in the RCT.

\clearpage
\pagenumbering{arabic}

\end{document}